\documentclass[conference,letterpaper]{IEEEtran}

\usepackage[utf8]{inputenc} 
\usepackage[T1]{fontenc}
\usepackage{url}
\usepackage{ifthen}
\usepackage{cite}
\usepackage[cmex10]{amsmath} 

\usepackage{setspace}
\usepackage{graphicx}
\usepackage{multicol}
\usepackage{multirow}
\usepackage{bm}
\usepackage{amsthm}
\usepackage{mathtools}
\usepackage{float}
\usepackage{amsrefs}
\usepackage{multirow}
\usepackage{amssymb}
\usepackage{amsfonts}
\usepackage{array}
\usepackage{tabu}
\usepackage{lettrine}
\usepackage{upgreek}
\usepackage{makecell}
\usepackage{nopageno}
\usepackage{color}
\numberwithin{equation}{section}
\newtheorem{theorem}{Theorem}
\newtheorem{lemma}{Lemma}
\usepackage{caption}
\usepackage{subcaption}
\usepackage{chngcntr}
\counterwithout{equation}{section}

\IEEEoverridecommandlockouts
\begin{document}
\title{Covert Communication in Continuous-Time Systems}

\author{%
  \IEEEauthorblockN{Ke Li}
  \IEEEauthorblockA{Electrical and Computer\\ Engineering,
                    UMass Amherst\\
                    Email: kli0@umass.edu}
  \and
  \IEEEauthorblockN{Don Towsley}
  \IEEEauthorblockA{College of Information and \\ Computer Science (CICS),
					UMass Amherst\\
					Email: towsley@cs.umass.edu}
  \and  
  \IEEEauthorblockN{Dennis Goeckel}
  \IEEEauthorblockA{Electrical and Computer\\ Engineergin,
		UMass Amherst\\
		Email: goeckel@ecs.umass.edu}
	
	\thanks{ This work was supported by the National Science Foundation under grant CNS-1564067.}
}

\maketitle

\begin{abstract}
	Recent works have considered the ability of transmitter Alice to communicate reliably to receiver Bob without being detected by warden Willie. These works generally assume a standard discrete-time model. But the assumption of a discrete-time model in standard communication scenarios is often predicated on its equivalence to a continuous-time model, which has not been established for the covert communications problem. Here, we consider the continuous-time channel directly and study if efficient covert communication can still be achieved. We assume that an uninformed jammer is present to assist Alice, and we consider additive white Gaussian noise (AWGN) channels between all parties. For a channel with  approximate bandwidth $W$, we establish constructions such that $\mathcal{O}(WT)$ information bits can be transmitted covertly and reliably from Alice to Bob in $T$ seconds for two separate scenarios: 1) when the path-loss between Alice and Willie is known; and 2) when the path-loss between Alice and Willie is unknown.
\end{abstract}


\section{Introduction}
Security is a major concern in modern wireless communications, where it is often obtained by encryption. However, this is not sufficient in applications where the very existence of the transmission arouses suspicion. For example, in military communications, the detection of a transmission may reveal activity in the region. Thus, it is important to study covert communication: hides the existence of the transmission, i.e., a transmitter (Alice) can reliably send messages to a legitimate receiver (Bob) without being detected by an attentive warden (Willie). 
Recent work studied the limits of reliable covert communications. Bash \textit{et al.} first studied such limits over discrete-time AWGN channels in \cite{LRC}, where a square-root law (SRL) is provided: Alice can transmit at most $\mathcal{O}(\sqrt{n})$ covert bits to Bob in $n$ channel uses of a discrete-time AWGN channel. This SRL was then established in successive work by Che \textit{et al.} in \cite{RDC1} over binary symetric channels (BSCs) and by Wang \text{et al.} in \cite{FLC} over arbitrary discrete memoryless channels (DMCs). The length of the secret key needed to achieve the SRL in covert communications over DMCs was established in \cite{CCN}. The work in \cite{FLC} and \cite{CCN} also established the scaling constants for the covert throughput. These works provide a thorough study in common discrete-time channel models when Willie has an accurate statistical characterization of the channel from Alice to him.

\begin{figure}[h!]
	\includegraphics[width=2.5in]{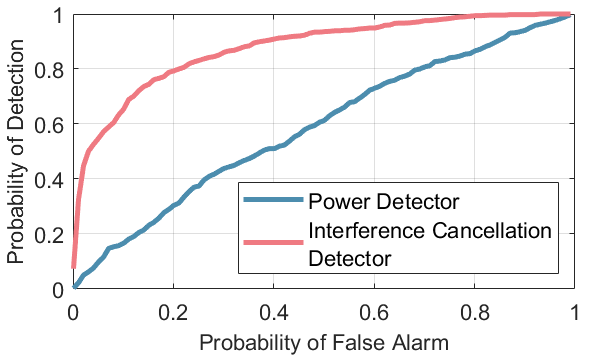}
	\centering
	\caption{Receiver operating characteristic curves of the interference cancellation detector and the standard power detector (implemented in a continuous-time covert communication system). The simulation is described in detail in Appendix A.}
	\label{ICD}
\end{figure}
In covert communications, Willie is attempting to determine whether he is just observing the background environment or a signal from Alice in that environment. Hence, uncertainty about that environment helps Alice to hide her transmission. Lee \textit{et al.}, \cite{ADC} and Che  \textit{et al.}, \cite{RDC2} show that $\mathcal{O}(n)$ covert bits in $n$ channel uses can be reliably transmitted from Alice to Bob if Willie is unsure of the variance of the noise at his receiver.
However, Goeckel \textit{et al.}, \cite{CCW} shows that Willie's lack of knowledge of his noise statistics can be compensated for by estimation through a collection of channel observations when Alice does not transmit. Thus, the limit of covert communications in this case goes back to the SRL. Sobers \textit{et al.} in \cite{CCP} then introduced another model to achieve positive covert rate: introducing an uninformed jammer to the system that randomly generates interference, hence providing the required uncertainty at Willie. In \cite{CCP}, it is proved that the optimal detector for Willie in the discrete-time model is a power detector. And, with Willie employing this optimal detector, Alice can covertly transmit $\mathcal{O}(n)$ bits in $n$ channel uses over both AWGN and block fading channels.

The works mentioned above are all based on a discrete-time model and thus implicitly assume that analogous results can be obtained on the corresponding continuous-time model. Bash \textit{et al.} first mentioned the potential fragility of such an assumption from \cite{LRC}: ideal $\text{sinc}(\cdot)$ pulse shapes are not feasible for implementation, perfect symbol synchronization might not always hold true, and sampling at higher rates sometimes has utility for signal detection at Willie even if the Nyquist ISI criterion is satisfied. In addition, continuous-time signals for transmission contain periodic features that can be extracted by the receiver to help it differentiate the signal from Gaussian noise. Thus, a power detector that is optimal at Willie \cite{CCP} in the discrete-time model may not be optimal in the continuous-time case. 

In \cite{CCC}, Sobers \textit{et al.} introduced a linear detector for the warden Willie that outperforms the standard power detector implemented in the continuous-time system in some limited scenarios. For general scenarios, we have developed an interference cancellation detector (inspired by co-channel interference cancellation techniques in cellular networks\cite{SC}), and show in Fig.~\ref{ICD} that this detector outperforms the standard power detector; hence, a major tenet of \cite{CCP} that facilitated the establishment of positive rate covert communications in the discrete-time case does not hold in the continuous-time case. Rather, Willie's detection capability benefits from the continuous-time setting, and hence raises questions on the covert limits in true continuous-time channels.
In this paper we will establish constructions for Alice such that positive covert rate is achievable. The reader will note how the constructions provided here are quite different from those in \cite{CCP}.

\section{System Model and Metrics}
\subsection{System Model}
Consider a scenario shown in Fig.~\ref{model} where transmitter Alice (``a'') wants to transmit a message to intended recipient Bob (``b'') reliably without being detected by a
warden Willie (``w''). A jammer (``j'') assists the communication by actively sending jamming signals, but without any coordination with Alice. 
\begin{figure}[h!]
	\includegraphics[width=2.5in]{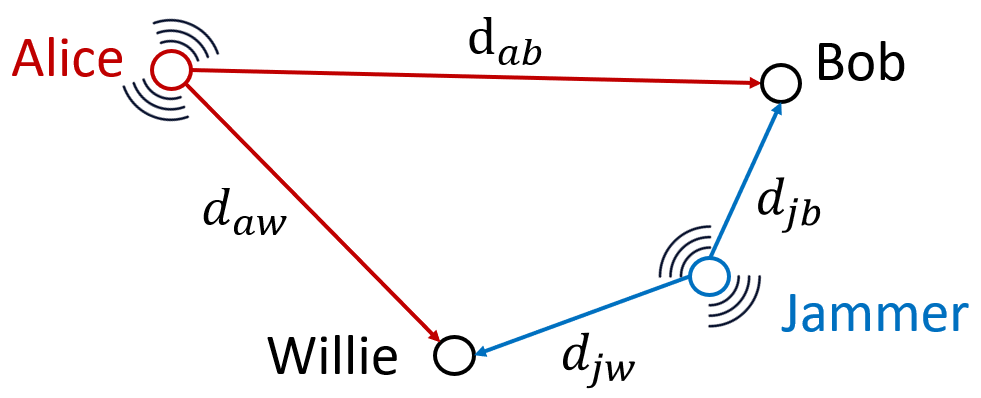}
	\centering
	\caption{System model: With help from a jammer, Alice attempts to transmit reliably and covertly to Bob in the presence of a warden Willie.}
	\label{model}
\end{figure}

For $t\in[0,T]$, we consider continuous-time channels where Alice and the jammer send symbols using pulse-shaped waveforms. Since Alice can send pulses at any time in the continuous time interval $[0,T]$, she and Bob share an infinite length key \cite{CCN} encoding those locations unknown to Willie.
If Alice decides to transmit, she maps her message to waveform $x_a(t)$ restricted to approximate bandwidth \footnote{See Appendix B for a discussion of the bandwidth of the constructions} $W$ under an average power constraint of $\sigma_a^2$.
The jammer transmits regardless if Alice transmitted or not. It sends waveform $x_j(t)$ that is also restricted to approximate bandwidth $W$ under an average power constraint of $\sigma_j^2$.
The channels between each transmitter and receiver pair are assumed to be AWGN, and thus the signal observed by Willie is given by:
\begin{align}
z(t)=
&\left\{ \begin{array}{lr} \frac{x_a(t-\tau_a)}{d_{aw}^{r/2}} + \frac{x_j(t-\tau_j)}{d_{jw}^{r/2}} +N^{(w)}(t), \,\text{Alice transmits} \\   \frac{x_j(t-\tau_j)}{d_{jw}^{r/2}} +N^{(w)}(t), \,\text{Alice does not transmit} \end{array} \right.
\label{z}
\end{align}
where $d_{xy}$ is the distance between a transmitter $x$ and a receiver $y$, $r$ is the path-loss exponent, $\tau_a$ and $\tau_j$ are time delays of Alice's and the jammer's signal, respectively, and $N^{(w)}(t)$ is the noise observed at Willie's receiver, which is a zero-mean stationary Gaussian random process with power spectral density $N_0^{(w)}/2$. 
Bob observes the channel output $y(t)$ at time $t$, which is analogous to $z(t)$ but with
the substitution of the noise $N^{(b)}(t)$ for $N^{(w)}(t)$, and $N^{(b)}(t)$ is a zero-mean stationary Gaussian random process with power spectral density of $N_0^{(b)}/2$.

We consider two scenarios: 1) the path-loss $d_{aw}^r$ between Alice and Willie is known; and 2) the path-loss $d_{aw}^r$ is unknown. In both scenarios, if not specified,
we assume the path-loss between any transmitter and receiver pair is one, without loss of generality.

\subsection{Metrics}
\subsubsection{Willie}
Based on his observations over the time interval, Willie attempts to determine whether Alice transmitted or not. We define the null hypothesis ($H_0$) as that Alice did not transmit during the time interval and the alternative hypothesis ($H_1$) as that Alice transmitted a message. We denote $P(H_0)$ and $P(H_1)$ as the probability that hypothesis $H_0$ or $H_1$ is true, respectively. Willie tries to minimize his probability of error $P_{e,w}=P(H_0)P_{FA}+P(H_1)P_{MD}$, where $P_{FA}$ and $P_{MD}$ are the probabilities of false alarm and missed detection at Willie, respectively. We assume that $P(H_0)$ and $P(H_1)$ are known to Willie. Since $P_{e,w} \geq \min(P(H_0),P(H_1))(P_{FA}+P_{MD})$ \cite{LRC}, we say that Alice achieves covert communication if, for a given $\epsilon >0$, $P_{FA}+P_{MD}\geq 1-\epsilon$ \cite{LRC}.

We assume that Willie has full knowledge of the statistical model: the time interval $[0,T]$, the parameters for Alice's random codebook generation, the parameters for the jammer's random interference generation, and the noise variance of his channel. Willie does not know the secret key shared between Alice and Bob, or the instantiation of the random jamming.

\subsubsection{Bob}
Bob should be able to reliably decode Alice's message. This is characterized by the probability $1-P_{e,b}$ where $P_{e,b}$ is the probability of error at Bob. We say that Alice achieves reliable communication if, for a given $\delta >0$, $P_{e,b}< \delta$ \cite{LRC}.

\section{Achievable Covert Communications: Known Path-Loss}

In this section, we consider the case that the path-loss between each transmitter and receiver pair is known, which we assume is one without loss of generality. We provide a construction for Alice and the jammer that consists of them sending randomly located pulses, and then demonstrate that the optimal detector for Willie, under this construction, is a threshold test on the number of pulses he observes. The ability for Alice to covertly send $\mathcal{O}(WT)$ bits is then established. This shows that covert communications with a positive rate can be achieved in continuous-time systems with equal path-loss.

\subsection{Construction}
We employ random coding arguments and generate codewords by independently drawing symbols from a zero-mean complex Gaussian distribution with variance $\sigma_a^2$. 
If Alice decides to transmit, she selects the codeword corresponding to her message, sets $f_i$ to the $i^{th}$ symbol of that codeword, and transmits the symbol sequence $\bold{f}=\{ f_1,f_2,\ldots \}$. 
The jammer transmits zero-mean complex Gaussian symbol sequence $\bold{v}=\{v_1.v_2,\ldots \}$, with variance $\sigma_j^2$. Here we choose $\sigma_j^2=\sigma_a^2$, i.e., Alice and the jammer use this same average transmit power, so that Alice can possibly hide her signal in the jammer's interference. 

Let $n=\lfloor WT\rfloor$ be an integer. Over the time interval $[0,T]$, Alice sends $M_a$ symbol pulses, where $M_a$ follows a binomial distribution with mean $\alpha n$, i.e., $M_a\sim B(n,\alpha)$, with constant $0\leq\alpha< 1$.
Alice's codeword length is chosen to be an integer close to $\alpha n-\epsilon_c n$ ($\epsilon_c$ is a small positive constant), such that as $n\to\infty$, there are enough pulses over $[0,T]$ for all of Alice's codeword symbols.
The jammer sends $M_j$ pulses, where $M_j$ follows a binomial distribution with mean $\beta n$, i.e., $M_j\sim B(n,\beta)$, with $\beta$ is uniformly distributed over $[\mu,\mu+\Delta]$, $0\leq\mu<\mu+\Delta\leq 1$ and $\Delta\geq\alpha$. Let $\tau_k, k=1,2,\ldots,M_a$ and $\tau_k',k=1,2,\ldots,M_j$ be be independent and identically distributed (i.i.d.) sequences of pulse delays of Alice and the jammer, respectively. The delays are drawn uniformly over $[0,T]$.
Alice's waveform within interval $[0,T]$ is then given by:
\begin{align}
x_a(t)=\sum_{k=1}^{M_a} f_k p(t-\tau_k)
\label{w1}
\end{align}
where $p(t)$ is a unit-energy pulse shaping filter with bandwidth $W$. Obviously, a waveform restricted to $[0,T]$ cannot have a finite bandwidth, we provide a brief discussion on this issue in Appendix B.
The jammer's waveform within interval $[0,T]$ is given by:
\begin{align}
x_j(t)=\sum_{k=1}^{M_j} v_k p(t-\tau'_k) .
\label{w2}
\end{align}
For an AWGN channel, Willie observes the signal $z(t)$ given in (\ref{z}). 

\subsection{Analysis}
To obtain an achievability result for covert communications, Willie should be assumed to employ an optimal detector. We will find an upper bound to the performance of that optimal detector by assuming a genie provides Willie additional information; in particular, we assume that Willie not only knows how the system is constructed (including $\alpha$, the distribution of $\beta$ and the transmission power $\sigma_a^2$ and $\sigma_j^2$ of the symbols), but also knows the number of pulses and the exact locations (timing) of each pulse on the channel in $[0,T]$. 
The only thing he does not know is from whom each pulse is sent.  In the next section, we will prove that the optimal test for Willie is a threshold test on the number of pulses he observed.

\subsection{Optimal Hypothesis Test}
Given the construction above, Willie's test is between the two hypotheses $H_0$ and $H_1$ where he has complete statistical knowledge of his observations when either hypothesis is true. We denote: Alice's decision on transmission as $D$ (which corresponds to hypothesis $H_0$ when Alice decides to transmit, or $H_1$ when she decides not to); the total number of pulses sent during time $T$ as $M$; the locations (over $[0,T])$ of the pulses as a vector $\bold{L}$; and the height (square root of the power) of the pulses as a vector $\bold{S}$, i.e., $\bold{S}$ is the vector of the original symbols sent. We want to first show that $M$ is a sufficient statistic for Willie's detection. The random variables $D$, $M$, $\bold{L}$ and $\bold{S}$ form a Markov chain shown in Fig.~\ref{markov1}, which illustrates the transition from Alice's state $D$ to Willie's  received signal $z(t)$.
The transitions of the Markov chain are:
\begin{itemize}
	\item{$D \longrightarrow M$: The conditional distribution of $M$, given $D$, is binomial with mean $\beta n$ when Alice does not transmit, and binomial with mean $\beta n +\alpha n$ when Alice transmits.}
	\item{$M \longrightarrow \bold{L},\bold{S}$: Given $M$, the distribution of $L_m$ for $m=1,2,\ldots, M$ is uniform over $[0,T]$. The distribution of $S_m$ for $m=1,2,\ldots,M$ is zero-mean Gaussian with variance $\sigma_a^2=\sigma_j^2$.}
\end{itemize}
\begin{figure}[h!]
	\includegraphics[width=2.2in]{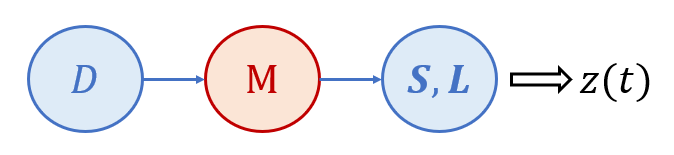}
	\centering
	\caption{Markov chain illustrating the transition from Alice's decision $D$ on transmission, to Willie's observed signal $z(t)$.}
	\label{markov1}
\end{figure}

Given the pulse locations and the height of the pulses, the signal $z(t)$ observed at Willie's receiver can be constructed from the pulse-shaping function $p(t)$ and the AWGN of Willie's channel. From the Markov chain shown in Fig.~\ref{markov1}, we see that $z(t)$ conditioned on $M$ is independent of $D$. Thus, $M$ is a sufficient statistic for Willie to make an optimal decision on Alice's presence. Therefore, by applying the Neyman-Pearson criterion, the optimal test for Willie to minimize his probability of error is the likelihood ratio test (LRT) \cite{FSS}:
\begin{align}
\Lambda(M=m)=\frac{P_{M|H_1}(m)}{P_{M|H_0}(m)} \overset{H_1}{\underset{H_0}{\gtrless}} \gamma
\label{lrt}
\end{align}
where $\gamma=P(H_0)/P(H_1)$, and $P_{M|H_1(m)}$ and $P_{M|H_0(m)}$ are the probability mass functions (pmfs) of the number of pulses given that Alice transmitted or did not transmit, respectively.
Given the LRT above, we want to show that this is equivalent to a threshold test on the number of pulses Willie observes at his receiver, which is true if the LRT exhibits monotonicity in $M$.
We employ the concept of stochastic ordering \cite{SO} to derive the desired monotonicity result. We say that $X$ is smaller than $Y$ in the likelihood ratio order (written as $X \leq_{lr} Y$) when $\frac{f_Y(x)}{f_X(x)}$ is non-decreasing over the union of their supports, where $f_Y(x)$ and $f_X(x)$ are pmfs or probability density functions (pdfs) of $Y$ and $X$, respectively \cite{CCP}. 
\begin{lemma} (Th. 1.C.11 in \cite{SO} adapted to pmfs):
	Consider a family of pmfs \{$g_b(\cdot), b\in\mathcal{X}$\} where $\mathcal{X}$ is a subset of the real line. Let $M(b)$ denote a random variable with pmf $g_b({\cdot})$.  For $\rho=0,1$,  let $B_\rho$ denote a random variable with support $\mathcal{X}$ and pdf $h_{B_\rho}(\cdot)$, and let
	$W_\rho=_d M(B_\rho)$ (where $=_d$ is defined as equality in distribution or law) denote a random variable with pmf given by:
	\begin{align}
	p_{W_\rho}(w)=\int_{b\in\mathcal{X}}  g_b(w) \, d \,h_{B_\rho}(b), \,\,\,\, w=0,1,\ldots \,. \nonumber
	\end{align}
	If $M(b) \leq_{lr} M(b')$ whenever $b\leq b'$, and if $B_0 \leq_{lr} B_1$, then:
	\begin{align}
	W_0 \leq_{lr} W_1. \nonumber
	\end{align}
\end{lemma}

We let:
\begin{align} 
b\overset{\Delta}{=}\left\{ \begin{array}{lr} \beta n, \,\text{when Alice does not transmit}\\   
\beta n +\alpha n,\, \text{when Alice transmits} \end{array} \right.\nonumber
\end{align}
and introduce two random variables $B_0$ and $B_1$ with pdfs given by:
\begin{align}
f_{B_\rho}(b)=	&\left\{ \begin{array}{lr} \frac{1}{\Delta}, \,\mu n<b\leq\mu n+\Delta n, \,\rho=0 \\   \frac{1}{\Delta}, \,\mu n+\alpha n<b\leq\mu n+\alpha n+\Delta n,\,\rho=1 \\
0,\, \text{else} \end{array} \right. \nonumber
\end{align}
The LRT in (\ref{lrt}) can be written as:
\begin{align}
\Lambda(M=m) 
=\frac{E_{B_1}\left[P_{M(b)}(m)\right]}{E_{B_0}\left[P_{M(b)}(m)\right]} \nonumber
\end{align}
where $M(b)$ follows a binomial distribution, i.e., $M(b)\sim B(n,\frac{b}{n})$.

\begin{theorem}
	Given the construction in the previous section, Willie's optimal detector compares the number of pulses he observes to a threshold.
\end{theorem}
\begin{proof}
	First, applying the definition of $\leq_{lr}$ to the densities of $B_0$ and $B_1$ yields that $B_0 \leq_{lr} B_1$. Then, let $R(m)\overset{\Delta}{=} \frac{P_{M(b')}(m)}{P_{M(b)}(m)}$, we write:
	\begin{align}
	R(m)&=\frac{\begin{pmatrix} 
		n\\m \end{pmatrix}\left(\frac{b'}{n}\right)^m \left(1-\frac{b'}{n}\right)^{n-m}}{\begin{pmatrix} 
		n\\m \end{pmatrix}\left(\frac{b}{n}\right)^m \left(1-\frac{b}{n}\right)^{n-m}}\nonumber \\
	&= \left(\frac{b'}{b}\right)^m \left(\frac{n-b'}{n-b}\right)^{n-m}  \nonumber
	\end{align}
	which monotonically increases as $m$ increases for $b\leq b'$. Thus, $M(b) \leq_{lr} M(b')$ whenever $b\leq b'$.
	The application of Lemma 1 then yields that $\Lambda(\cdot)$ is non-decreasing in $m$. Therefore, the LRT is equivalent to the test:
	\begin{align}
	M  \overset{H_1}{\underset{H_0}{\gtrless}} \gamma'
	\label{test}
	\end{align}
	corresponding to a threshold test on the number of pulses observed by Willie.	
\end{proof}

Dividing both sides of (\ref{test}) by $n$ yields the equivalent test:
\begin{align}
\frac{M}{n}\overset{H_1}{\underset{H_0}{\gtrless}} \gamma_n \nonumber
\end{align}
where $\gamma_n=\gamma'/n$. For any finite $n$, there is an optimal threshold $\gamma_n$ such that it minimizes Willie's probability of error in detecting Alice's existence. However, we will show that for any $\gamma_n$ Willie chooses, he will not be able to detect Alice as $n\to\infty$; that is, for any $\epsilon>0$, there exists a construction such that $P_{FA}+P_{MD}>1-\epsilon$ for $n$ large enough.

\subsection{Covert Limit}
Recall that $\beta$ is a uniform random variable on $[\mu, \mu+\Delta]$, with $0\leq\mu<\mu+\Delta\leq 1$ and constant $\Delta\geq\alpha$. Let $P_{FA}(u)$ and $P_{MD}(u)$ be Willie's probability of false alarm and missed detection conditioned on $\beta=u$, respectively. Then:
\begin{align}
P_{FA}(u)&=P\left(\frac{M}{n} \geq \gamma_n \mid H_0\right) =P\left( \frac{1}{n}\sum_{i=1}^{n} R_i \geq \gamma_n \mid H_0 \right) \nonumber
\end{align}
where, under $H_0$, $R_i$ is a Bernoulli random variable that takes value one with probability $u$. By the weak law of large numbers, $\frac{1}{n}\sum_{i=1}^{n} R_i$ converges in probability to $u$. Thus, for any $\eta >0$, there exists $N_0$ such that, for $n\geq N_0$, we have:
\begin{align}
P\left( \frac{1}{n}\sum_{i=1}^{n} R_i \in \left( u-\eta,u+\eta \right) \mid H_0 \right) >1-\frac{\epsilon}{2} \,.\nonumber
\end{align}
Therefore, for $n>N_0$, $P_{FA}(u)> 1-\frac{\epsilon}{2}$ for any $\gamma_n < u-\eta$.
Analogously, we write:
\begin{align}
P_{MD}&=P\left(\frac{M}{n} \leq \gamma_n \mid H_1\right) =P\left( \frac{1}{n}\sum_{i=1}^{n} R_i \leq \gamma_n \mid H_1\right) \nonumber
\end{align}
Likewise, by the weak law of large numbers, for any $\eta>0$, there exists $N_1$ such that, for $n\geq N_1$, we have:
\begin{align}
P\left( \frac{1}{n}\sum_{i=1}^{n} R_i \in \left( u+\alpha-\eta,u+\alpha+\eta \right) \mid H_1 \right) >1-\frac{\epsilon}{2} \,. \nonumber
\end{align}
Therefore, for $n>N_1$, $P_{MD}(u)>1-\frac{\epsilon}{2}$ for any $\gamma_n>\alpha+u+\eta$.

Define the set $\mathcal{A}=\{u:u-\eta<\gamma_n<\alpha+u+\eta\}$. We have established that, for any $u\in \mathcal{A}^c$ and  $n>\max(N_0,N_1)$, $P_{FA}(u)+P_{MD}(u) > 1-\frac{\epsilon}{2}$. The probability of $\mathcal{A}$ has the following upper bound:
\begin{align}
P(A)&=P(\gamma_n-\alpha-\eta < \beta <\gamma_n+\eta) \leq \frac{2\eta+\alpha}{\Delta} \,. \nonumber
\end{align}
By choosing $\eta=\frac{\epsilon\Delta}{4}$ and $\alpha=\frac{\epsilon\Delta}{2}$, we have $P(\mathcal{A})\leq \frac{\epsilon}{2}$, i.e., $P(\mathcal{A}^c)>1-\frac{\epsilon}{2}$. Hence,
\begin{align}
P_{FA}+P_{MD}&=E_\beta[P_{FA}(\beta)+P_{MD}(\beta)] \nonumber \\
& \geq E_\beta[P_{FA}(\beta)+P_{MD}(\beta) \mid \mathcal{A}^c]P(\mathcal{A}^c) \nonumber \\
& > 1-\frac{\epsilon}{2} \,.\nonumber
\end{align}
Thus, Alice can send an average of $\alpha n=\frac{\epsilon\Delta}{2}\lfloor WT \rfloor$ pulses with a constant power and remain covert from Willie. Note that since the maximum interference from the jammer at Bob can be upper bounded by a constant, reliability is also achieved under the same construction. Therefore, $\mathcal{O}(WT)$ bits can be transmitted covertly and reliably from Alice to Bob.

\section{Achievable Covert Communications: Unknown Path-Loss}
In this section, we consider the case that the jammer does not know the exact path-loss between Alice and Willie, but only knows an upper and lower bound of the received power from Alice at Willie. Without loss of generality, we assume the path-loss between the jammer and Willie is one.  Since the jammer does not know the exact path-loss between Alice and Willie, it cannot use a power that results in the pulses of Alice and the jammer arriving at Willie with the same power as in the previous section.  With the construction of the previous section, Willie could separate Alice and the jammer by looking for a pulse power distribution that is the combination of two distributions.  Therefore, to prevent Willie from detecting Alice, another construction is needed. The idea of the construction is to let the jammer send pulses with multiple power levels that cover
a wide range of the power spectrum, so that if Alice uses an average power within that range, she can possibly hide herself in the jammer’s interference. We will establish the construction and show that under such construction, covert communications with a positive covert rate can be achieved.

\subsection{Construction} 
\subsubsection{Alice}
Similar to before, we employ random coding arguments and generate codewords by independently drawing symbols from a zero-mean complex Gaussian distribution. However, Alice's average transmission power is random: she chooses a power level uniformly over $[P_a,P_a+\Delta_{P_a}]$ where $P_a$ and $\Delta_{P_a}$ are constants  ($P_a+\Delta_{P_a}\leq\sigma_a^2$), and transmits symbols with this power over $[0,T]$. Alice transmits a total of $M_n=\lfloor \alpha n \rfloor$ pulses ($0\leq\alpha<1$ is a constant) over $[0,T]$. Therefore, Alice's waveform within $[0,T]$ is given by:
\begin{align}
x_a(t)=\sum_{i=1}^{M_n} f_i p(t-\tau_i)
\nonumber
\end{align}
where $f_i, i=1,2,\ldots,M_n$ is a sequence of i.i.d zero-mean Gaussian symbols with the same variance that is uniformly drawn from $[P_a,P_a+\Delta_{P_a}]$, and $\tau_i, i=1,2,\ldots,M_n$ is a sequence of i.i.d pulse delays that are uniformly distributed in $[0,T]$.

\subsubsection{Jammer}
The jammer also sends i.i.d. Gaussian symbols. It first determines a number $K$ of power levels according to a Poisson distribution, i.e., $K\sim Pois\left(\lambda_j\right)$, where $\lambda_j$ is a constant. It then chooses each of the $K$ power levels uniformly in $[P_j,P_j+\Delta_{P_j}]$, where $P_j$ and $\Delta_{P_j}$ are constants ($P_j+\Delta_{P_j}\leq\sigma_j^2$), to transmit its symbols. Note that the range of the jammer's power at Willie needs to cover the range of all possible values of Alice's power at Willie, and since the jammer knows an upper and lower bound of the received power from Alice and Willie, $P_j$ and $\Delta_{P_j}$ are chosen such that $\left[\frac{P_a}{d_{aw}^r},\frac{P_a+\Delta_{P_a}}{d_{aw}^r}\right] \subset \left[P_j,P_j+\Delta_{P_j}\right]$. Note that this implies that the jammer knows a lower bound on the distance between Alice and Willie.
The jammer transmits $M_n$ number of pulses for each power level it chooses. Hence, it will transmit a total of $KM_n$ pulses over $[0,T]$.  Therefore, the jammer's waveform is given by:
\begin{align}
x_j(t)=\sum_{k=1}^{K}\sum_{i=1}^{M_n} v_{i,k} p(t-\tau_{i,k}')
\nonumber
\end{align}
where $v_{i,k}, i=1,2,\ldots,M_n, k=1,2,\ldots,K$ is a sequence of i.i.d zero-mean Gaussian symbols with variance being the $k^{\text{th}}$ power level randomly chosen by the jammer, and $\tau_i, i=1,2,\ldots,M_n, k=1,2,\ldots,K$ is a sequence of i.i.d pulse delays that are uniformly distributed in $[0,T]$.

\subsection{Analysis}
For achievability, we derive an upper bound to the performance of Willie's optimal detector by assuming a genie provides Willie extra knowledge on the exact power range $\left[\frac{P_a}{d_{aw}^2},\frac{P_a+\Delta_{P_a}}{d_{aw}^r}\right]$ received from Alice, the distribution of the number of the jammer's power levels (including all of the parameters), and the values of all power levels employed by the jammer and Alice (if she decided to transmit), but not which power level is employed by whom.

\subsection{Optimal Hypothesis Test}
In this section, we show that the number of power levels in the range of $\left[\frac{P_A}{d_{aw}^2},\frac{P_A+\Delta_{P_A}}{d_{aw}^r}\right]$, which we term the detection region, is a sufficient statistic for Willie in deciding between hypothesis $H_0$ or $H_1$. Fig.~\ref{power_level} illustrates the power levels received at Willie. 
\begin{figure}[h!]
	\includegraphics[width=3.3in]{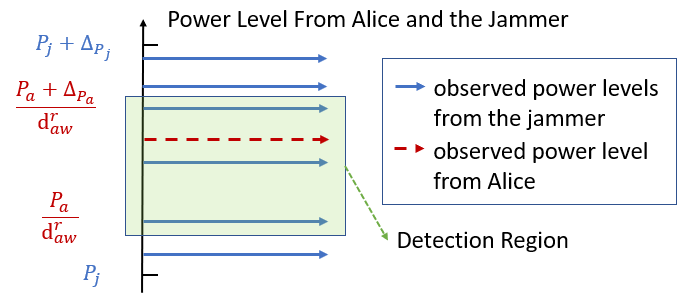}
	\centering
	\caption{Willie's received power levels from Alice and the jammer. An impulse means a power level Alice or the jammer chooses for transmission.}
	\label{power_level}
\end{figure}

Let $K_1$ be the number of power levels inside the detection region, $K_2$ be the number of power levels outside the detection region, i.e. $K=K_1+K_2$. By construction, all of the power levels sent by the jammer form a Poisson point process with $K\sim Pois\left(\lambda_j\right)$ on $[P_j,P_j+\Delta_{P_j}]$. Note that for a Poisson point process, generating $K$ power levels with mean $\lambda_j$ and placing them uniformly over $[P_j,P_j+\Delta_{P_j}]$ is equivalent to: generating $K_1$ power levels with mean $\frac{\Delta_{P_a}}{\Delta_{P_j}d_{aw}^r} \lambda_j$ and placing them uniformly inside the detection region, and generating $K_2$ power levels with mean $\left(1-\frac{\Delta_{P_a}}{\Delta_{P_j}d_{aw}^2}\right) \lambda_j$ and placing them uniformly outside the detection region. This is critical in the proof below.

Recall that $D$ denotes Alice's decision on transmission, $L$ denotes the locations (in $[0,T]$) of all of the pulses sent, and $\bold{S}$ denotes the height of the pulses. We also denote the values of all power levels (within and outside the detection region) as a vector $\bold{V}$.
The random variables $D$, $K_1$, $\bold{V}$, $\bold{L}$ and $\bold{S}$  form a Markov chain shown in Fig.~\ref{markov}, which illustrates the transition from Alice's state $D$ to Willie's received signal $z(t)$. The transitions of the Markov chain are:
\begin{itemize}
	\item{$D \longrightarrow K_1$:
	$K_1$ and $K_1-1$ are characterized by a Poisson process with mean $\frac{\Delta_{P_a}\lambda_j}{\Delta_{P_j} d_{aw}^r}$ when Alice does not transmit, and a Poisson process with mean $\frac{\Delta_{P_a}\lambda_j}{\Delta_{P_j} d_{aw}^r}$ when she does transmit.}
	\item{$K_1 \longrightarrow \bold{V}, \bold{L}$: Let $V_k, k=1,2,\ldots,K_1$ be the values of power levels within the detection region, and $V_k, k=K_1+1,K_1+2,\ldots,K$ the power values outside the detection region. Given $K_1$, the conditional distribution of $V_k, k=1,2,\ldots,K_1$, is uniform within the detection region. Note that $K_2$ is independent of $D$ since the pulses sent with power levels outside the detection region can only come from the jammer, no matter if Alice transmits or not.  Given $K_2$ (Poisson with mean $\left(1-\frac{\Delta_{P_A}}{\Delta_{P_J}d_{aw}^2}\right) \lambda_j$), the distribution of 
	$V_k, k=K_1+1,K_1+2,\ldots,K$, is uniform outside the detection region.
	Let $\{L_{k,m}: k=1,\ldots,K_1, m=1,\ldots,M_n\}$ denote the locations (in $[0,T]$) of pulses sent with power within the detection region, and $\{L_{k,m}: k=K_1+1,\ldots,K, m=1,\ldots,M_n\}$ denote the locations of pulses sent with power outside the detection region.
	Given $K_1$, the distribution of $L_k,m$ for $k=1,2,\ldots,K_1$ and all $m$ is uniform over $[0,T]$. Given $K_2$, the distribution of $L_k$ for $k=K_1+1,K_1+2,\ldots,K$ and all $m$ is also uniform over $[0,T]$, which is independent from $D$.}
	\item{$\bold{V}, \bold{L} \longrightarrow \bold{S}, \bold{L}$: The conditional distribution of $S_{k,m}$, for $k=1,2,\ldots,K, m=1,2,\ldots, M_n$, given $V_k$, is a zero-mean Gaussian random variable with variance $V_k$.} 
\end{itemize}
\begin{figure}[h!]
	\includegraphics[width=2.6in]{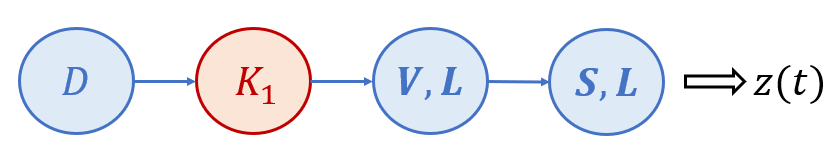}
	\centering
	\caption{Markov chain illustrating the transition from Alice's decision $D$ on transmission, to Willie's observed signal $z(t)$.}
	\label{markov}
\end{figure}

Given the pulse locations and the height of the pulses, the signal $z(t)$ can be  constructed from $p(t)$ and the AWGN of Willie's channel. From the Markov chain shown in Fig.~\ref{markov}, we see that  $z(t)$ conditioned on $K_1$ is independent of $D$. Therefore, $K_1$ is a sufficient statistic for Willie to decide between hypotheses $H_0$ and $H_1$.

In particular, hypotheses $H_0$ and $H_1$ can be characterized as:
\begin{itemize}
	\item{$H_0$: the number of power levels within the detection region follows $Pois\left(\frac{\lambda_j\Delta_{P_a}}{\Delta_{P_j} d_{aw}^r}\right)$;}	
	\item{$H_1$: the number of power levels within the detection region follows $Pois\left(\frac{\lambda_j\Delta_{P_a}}{\Delta_{P_j} d_{aw}^r}\right) + 1$.}
\end{itemize}

\subsection{Covert Limit}
Let $P_0$ and $P_1$ denote the distribution of the number of power levels observed by Willie (within Willie's detection range) given $H_0$ and $H_1$, respectively:
\begin{align}
P_0(k)=\frac{\lambda^k e^{-\lambda}}{k!} , \,\,\, k\geq 0
\label{p0}
\end{align}  
and
\begin{align}
P_1(k)=\frac{\lambda^{k-1} e^{-\lambda}}{(k-1)!} , \,\,\, k\geq 1
\label{p1}
\end{align}
where $\lambda=\frac{\lambda_j \Delta_{P_a}}{\Delta_{P_j} d_{aw}^r}$.
Theorem 13.1.1 in \cite{TSH} shows that for the optimal hypothesis test,
\begin{align}
P_{FA}+P_{MD}=1-\mathcal{V}_T (P_0,P_1) \nonumber
\end{align}
where
\begin{align}
\mathcal{V}_T (P_0,P_1)=\frac{1}{2} \sum_k |P_0(k)-P_1(k)| \nonumber
\end{align}
is the total variation distance between $P_0$ and $P_1$, where the sum is over all $k$ in the support of $P_0 \cup P_1$. Therefore, by the definition of covertness, if
\begin{align}
\mathcal{V}_T (P_0,P_1) \leq \epsilon ,
\end{align} 
Alice achieves covert communications.

Given (\ref{p0}) and (\ref{p1}), we derive:
\begin{align}
&\mathcal{V}_T (P_0,P_1) = \frac{1}{2}\sum_{k=1}^{\infty}|P_0(k)-P_1(k)|+\frac{1}{2} P_0(0) \nonumber \\
&= \frac{1}{2} \sum_{k=1}^{\infty} \frac{\lambda^{k-1} e^{-\lambda}}{(k-1)!} \left|\frac{\lambda}{k}-1 \right| + \frac{1}{2} e^{-\lambda} \nonumber \\
&= \frac{1}{2} \Bigg[ \sum_{k=1}^{\lambda} \frac{\lambda^{k-1}e^{-\lambda}}{(k-1)!} \left(\frac{\lambda}{k}-1\right) + \sum_{k=\lambda+1}^{\infty} \frac{\lambda^{k-1}e^{-\lambda}}{(k-1)!} \left(1-\frac{\lambda}{k}\right) \nonumber \\
&\qquad\qquad\qquad\qquad\qquad\qquad\qquad\qquad\qquad\qquad\quad\,\,\,\,\, + e^{-\lambda} \Bigg] \nonumber \\
&= \frac{1}{2} \Bigg[ \sum_{k=1}^{\lambda} \frac{\lambda^k e^{-\lambda}}{k!} - \sum_{k=1}^{\lambda} \frac{\lambda^{k-1}e^{-\lambda}}{(k-1)!} + \sum_{k=\lambda+1}^{\infty} \frac{\lambda^{k-1}e^{-\lambda}}{(k-1)!} \nonumber \\
&\qquad\qquad\qquad\qquad\qquad\qquad\qquad\quad - \sum_{k=\lambda+1}^{\infty} \frac{\lambda^k e^{-\lambda}}{k!} + e^{-\lambda} \Bigg] \nonumber \\
&= \frac{1}{2} \Bigg[ \sum_{k=1}^{\lambda}\frac{\lambda^k e^{-\lambda}}{k!} + \sum_{k=\lambda}^{\infty}\frac{\lambda^k e^{-\lambda}}{k!}  -  \sum_{k=0}^{\lambda-1} \frac{\lambda^k e^{-\lambda}}{k!} \nonumber \\
&\qquad\qquad\qquad\qquad\qquad\qquad\qquad\quad - \sum_{k=\lambda+1}^{\infty} \frac{\lambda^k e^{-\lambda}}{k!} + e^{-\lambda} \Bigg]  \nonumber \\
&= \frac{1}{2} \left( \sum_{k=1}^{\infty} \frac{\lambda^k e^{-\lambda}}{k!} + \frac{\lambda^\lambda e^{-\lambda}}{\lambda!} - \sum_{k=0}^{\infty} \frac{\lambda^k e^{-\lambda}}{k!} + \frac{\lambda^\lambda e^{-\lambda}}{\lambda!} + e^{-\lambda} \right) \nonumber \\
&= \frac{\lambda^\lambda e^{-\lambda}}{\lambda!} \nonumber
\end{align}
Using Stirling's approach, this can be upper bounded as:
\begin{align}
\frac{\lambda^\lambda e^{-\lambda}}{\lambda!} & \leq \frac{\lambda^\lambda e^{-\lambda}}{\sqrt{2\pi}\lambda^{\lambda+1/2}e^{-\lambda}}  = \frac{1}{\sqrt{2\pi\lambda}} \,. \nonumber
\end{align} 
Thus, if 
\begin{align}
\lambda \geq \frac{1}{2\pi\epsilon^2} , \nonumber
\end{align}
i.e.,
\begin{align}
\lambda_j \geq \frac{\Delta_{P_j} d_{aw}^r}{2\pi\Delta_{P_a}\epsilon^2}\, , 
\label{lj}
\end{align}
covertness is achieved. This implies that one of the two strategies can be employed: 1) Alice chooses a $\Delta_{P_a}$ and the jammer can use an upper bound on $d_{aw}^r$ to choose $\lambda_j$; 2) the jammer chooses a $\lambda_j$ and Alice can use $d_{aw}^r$ to choose $\Delta_{P_a}$.

Since the maximum interference from the jammer at Bob can be upper bounded by a constant, reliability is achieved under the same construction.
Thus, under this construction, Alice can achieve covert and reliable communications when the path-loss between her and Willie is unknown.
Also, since under the above construction, Alice can send $M_n=\lfloor\alpha\lfloor WT \rfloor\rfloor$ pulses with a constant power (which does not decrease with $WT$), $\mathcal{O}(WT)$ bits can be transmitted covertly and reliably from Alice to Bob.

\section{Conclusion}
In this paper, we have studied covert communications in continuous-time systems, where Alice wants to reliably communicate with Bob in the presence of a jammer without being detected by Willie. We established constructions that allow Alice to achieve covert communications in both cases when the path-loss between Alice and Willie is known and unknown. We proved that $\mathcal{O}(WT)$ covert information bits on a channel with approximate bandwidth $W$ can be reliably transmitted from Alice to Bob in $T$ seconds for both cases. In this paper, an infinite number of key bits shared between Alice and Bob is needed. A direction for future work is to consider the use of a finite number of key bits and the values of the scaling constants.

\begin{appendices}
	
	\section{Simulation of Fig.~\ref{ICD}}
	\cite{SC} introduces a co-channel interference cancellation technique with initial signal separation when the signals have different timing offsets. Here we apply similar techniques in covert communication systems where the receiver only wants to detect the existence of the power -- a single bit of information, instead of a signal from its mixture of another signal. In the simulation, we set the number of trials to $1000$. We let Alice and the jammer send $200$ i.i.d zero-mean Gaussian symbols with pulse-shaped waveforms (using square-root raised cosine pulse shaping filter with roll-off factor $0.2$). The two signals have symbol period $T_s=48$ discrete-time samples and time delay difference $T_s/6$. Alice's signal to noise ratio (SNR) is set to be $5$ dB, and the jammer's SNR is set to be $20$ dB. 
	The jammer's signal is treated as interference and is subtracted using the same techniques in \cite{SC} without iteration.
	Standard power detector is then applied at the output to detect Alice's presence.

	\section{Discussion of the Bandwidth of the Constructions}
	Here we provide a brief discussion on the bandwidth of our 
	construction.  For either of our constructions, each of a random  or a constant number 
		$M$ of pulses with pulse shape $p(t)$ is multiplied by its corresponding 
	symbol and then placed with delay randomly drawn from the interval 
	$[0,T]$.  This results in a waveform:
	\begin{align}
		X(t) = \sum_{k=1}^N a_k p(t - \tau_k) \nonumber
	\end{align}
	where $a_k, n=1,2,\ldots,N$ is the sequence of zero-mean 
	independent symbol values, and $\tau_k, k=1,2,\ldots, M$ is the i.i.d. 
	sequence of pulse delays.   Since the delays are drawn uniformly over 
	only the interval $[0,T]$, the process $X(t)$ is not wide-sense 
	stationary and thus its bandwidth is not strictly defined.  Hence, 
	consider rather the following random process, which is an extension of 
	the construction to the infinite interval:
	\begin{align}
		\tilde{X}(t) = \sum_{i=-\infty}^{\infty} \sum_{k=1}^N a_k^{(i)} p (t - 
		\tau_k^{(i)} - i T) \nonumber
	\end{align}
	where $a_k^{(0)}= a_k$ and $\tau_k^{(0)} = \tau_k$, 
	$k=1,2,\ldots,N$, and the values for the intervals outside of $[0,T]$ 
	are chosen independently but according to the same construction as 
	within $[0,T]$.  The random process $\tilde{X}(t)$ is wide-sense 
	stationary, and, through standard digital communication system analysis 
	arguments, has power spectral density $S_{\tilde{X}}(f) = |P(f)|^2$, 
	where $P(f)$ is the Fourier transform of $p(t)$.  Hence, the bandwidth 
	of $\tilde{X}(t)$ is the same as that of $P(f)$.  Observing that $X(t)$ 
	is a windowed version of $\tilde{X}(t)$ and that $T$ is very large, the 
	signal $X(t)$ is approximately bandlimited to the bandwidth $W$ of $p(t)$.
	
\end{appendices}

\end{document}